\documentclass[runningheads]{llncs}

\usepackage{amsmath,amssymb}
\usepackage{mathtools}
\usepackage{graphicx}

\usepackage{wasysym} 

\usepackage[vlined,linesnumbered,ruled,titlenotnumbered]
{algorithm2e}

\usepackage{tikz}
\usetikzlibrary{tikzmark}
\tikzset{every picture/.style={remember picture}}

\usepackage{csquotes}
\usepackage{booktabs}
\usepackage{xspace}
\usepackage{xargs}
\usepackage{subfigure}
\usepackage{colortbl}
\usepackage{diagbox}
\usepackage{hyperref}

\newcommand{\ignore}[1]{}

\newcommand{\qedhere}{\ensuremath{\Box}}

\newcommand{\KP}{KP\xspace}
\newcommand{\SHARPOPTKNAPSACK}{\#KNAPSACK$^*$\xspace}
\newcommand{\SHARPKNAPSACK}{\#KNAPSACK\xspace}
\newcommand{\KPFEASIBLE}{\mathcal{S}\xspace}
\newcommand{\KPOPT}{\mathcal{S}^{*}\xspace}

\title{Exact Counting and Sampling of Optima for the Knapsack Problem}

\author{
Jakob Bossek\inst{1}\orcidID{0000-0002-4121-4668} \and
Aneta Neumann\inst{2}\orcidID{0000-0002-0036-4782} \and
Frank Neumann\inst{2}\orcidID{0000-0002-2721-3618}
}
\authorrunning{J. Bossek et al.}

\institute{
Statistics and Optimization, University of M{\"u}nster, Germany \\
\and
School of Computer Science, The University of Adelaide, Australia \\
\email{bossek@wi.uni-muenster.de, \{aneta,frank\}.neumann@adelaide.edu.au}}

\begin{document}

\maketitle

\begin{abstract}
Computing sets of high quality solutions has gained increasing interest in recent years. In this paper, we investigate how to obtain sets of optimal solutions for the classical knapsack problem. We present an algorithm to count exactly the number of optima to a zero-one knapsack  problem instance. In addition, we show how to efficiently sample uniformly at random from the set of all global optima. In our experimental  study, we investigate how the number of optima develops for classical random benchmark instances dependent on their generator parameters. We find that the number of global optima can increase exponentially for practically relevant classes of instances with correlated weights and profits which poses a justification for the considered exact counting problem.
\end{abstract}


\section{Introduction}

Classical optimisation problems ask for a single solution that maximises or minimises a given objective function under a given set of constraints. This scenario has been widely studied in the literature and a vast amount of algorithms are available. In the case of NP-hard optimisation problems one is often interested in a good approximation of an optimal solution. Again, the focus here is on a single solution.

Producing a large set of optimal (or high quality) solutions allows a decision maker to pick from structurally different solutions. Such structural differences are not known when computing a single solution. Computing a set of optimal solutions has the advantage that more knowledge on the structure of optimal solutions is obtained and that the best alternative can be picked for implementation. As the number of optimal solutions might be large for a given problem, sampling from the set of optimal solutions provides a way of presenting different alternatives.

Related to the task of computing the set of optimal solutions, is the task of computing diverse sets of solutions for optimisation problems. This area of research has obtained increasing attention in the area of planning where the goal is to produce structurally different high quality plans~\cite{DBLP:conf/ecai/SohrabiRUH16,DBLP:conf/aips/0001SUW18,DBLP:conf/aaai/KatzSU20,DBLP:conf/aaai/0001S20}. Furthermore, different evolutionary diversity optimisation approaches which compute diverse sets of high quality solutions have been introduced~\cite{ulrich2011maximizing,neumann2018discrepancy,neumann2019evolutionary}. For the classical traveling salesperson problem, such an approach evolves a diverse set of tours which are all a good approximation of an optimal solution~\cite{DBLP:conf/gecco/DoBN020}.

Counting problems are frequently studied in the area of theoretical computer science and artificial intelligence~\cite{10.5555/3202481,DBLP:journals/cc/FournierMM15,DBLP:conf/aaai/FichteHM19,DBLP:conf/sat/FichteHMW18}. Here the classical goal is to count the number of solutions that fulfill a given property.
This might include counting the number of optimal solutions. Many counting problems are \#P-complete~\cite{journals/tcs/Valiant79} and often approximations of the number of such solutions, especially approximations on the number of optimal solutions are sought~\cite{DBLP:conf/approx/DyerGGJ00}. Further examples include counting the number of shortest paths between two nodes in graphs~\cite{MihalakMSRW2016ApproxCountingShortestPath} or exact counting of minimum-spanning-trees~\cite{BroderM1997CountingMSTs}.
For the knapsack problem~(KP), the problem of counting the number of feasible solutions, i.e. the number of solutions that do not violate the capacity constraint, is \#P-complete. As a consequence, different counting approaches have been introduced to approximately count the number of feasible solutions~\cite{Dyer2003ApproximateCounting,DBLP:journals/iandc/RizziT19,Stefankovic2012FPTASKnapsackCounting}.

In this paper, we study the classical zero-one knapsack problem (KP). We develop an algorithm that is able to compute all optimal solutions for a given knapsack instance. The algorithm adapts the classical dynamic programming approach for KP in the way that all optimal solutions are produced implicitly. As the number of such solutions might grow exponentially with the problem size for some instances, we develop a sampling approach which samples solutions for the set of optimal solutions uniformly at random.

We carry out experimental investigations for different classes of knapsack instances given in the literature (instances with uniform random weights and instances with different correlation between weights and profits). Using our approach, we show that the number of optimal solutions significantly differs between different knapsack instance classes. In particular, for instances with correlated weights and profits -- a group of great importance in practical applications -- an exponential growth of optima is observed.
In addition, we point out that changing the knapsack capacity slightly can reduce the number of optimal solutions from exponential to just a single solution. 

The paper is structured as follows. In the next section, we introduce the task of computing the set of optimal solutions for the knapsack problem. Afterwards, we present the dynamic programming approach for computing the set of optimal solutions and show how to sample efficiently from the set of optimal solutions without having to construct the whole set of optimal solutions. In our experimental investigations, we show the applicability of our approach to a wide range of knapsack instances and point out insights regarding the number of optimal solutions for these instances. Finally, we finish with some concluding remarks.

\section{Problem Formulation}
\label{sec:sec02}

We now introduce the problem of computing all optimal solutions for the knapsack problem. In the following, we use the standard notation $[n]=\{1,\ldots,n\}$ to express the set of the first $n$ positive integers.
The problem studied is the classical NP-hard \emph{zero-one knapsack problem} (\KP). We are given a knapsack with integer capacity $W>0$ and a finite set of $n$ items, each with positive integer weight $w_i$ and associated integer profit or value $v_i$ for $i \in [n]$. Each subset $s \subset [n]$ is called a \emph{solution/packing}. We write
\begin{align*}
    w(s) = \sum_{i \in s} w_i \text{ and } v(s) = \sum_{i \in s} v_i
\end{align*}
for the total weight and value respectively. A solution is \emph{feasible} if its total weight does not exceed the capacity. Let 
\begin{align*}
    \KPFEASIBLE = \{s \mid s \subset [n] \wedge w(s) \leq W\}
\end{align*}
be the set of feasible solutions. The goal in the optimisation version of the problem is to find a solution $s^{*} \in \KPFEASIBLE$ such that 
\begin{align*}
    s^{*} = \text{arg$\,$max}_{s \in \KPFEASIBLE} v(s).
\end{align*}
Informally, in the optimisation version of the \KP, we strive for a subset of items that maximises the total profit under the constraint that the total weight remains under the given knapsack capacity.

Let $v_{\max}$ be the value of an optimal solution $s^{*}$ and let \begin{align*}
\KPOPT = \left\{s \,\mid\, s \in \KPFEASIBLE \wedge v(s) = v_{\max}\right\}
\end{align*}
be the set of optimal solutions. 
In this work we study a specific counting problem which we refer to as \SHARPOPTKNAPSACK in the following. Here, the goal is to determine \emph{exactly} the cardinality of the set $\KPOPT$. Note that this is a special case of the classic counting version \SHARPKNAPSACK where we aim to count the set of all feasible solutions $\KPFEASIBLE$, and $\KPOPT \subset \KPFEASIBLE$. In addition, we are interested in procedures to sample uniformly at random a subset of $k$ out of $|\KPOPT|$ solutions with $k \leq |\KPOPT|$.

\section{Exact Counting and Sampling of Optima}
\label{sec:sec03}

In this section we introduce the algorithms for the counting and sampling problems stated. We first recap the classic dynamic programming algorithm for the zero-one KP as it forms the foundation for our algorithm(s).

\subsection{Recap: Dynamic Programming for the KP}

Our algorithms are based on the dynamic programming approach for the optimisation version (see, e.g. the book by Kellerer et al.~ \cite{Kellerer2004KnapsackProblems}). This well-known algorithm maintains a table $V$ with components $V(i,w)$ for $0 \leq i \leq n, 0 \leq w \leq W$. Here, component $V(i,w)$ holds the maximum profit that can be achieved with items up to item $i$, i.e., $\{1, \ldots, i\}$, and capacity $w$. The table is constructed bottom-up following the recurrence
\begin{align*}
V(i,w) = \max\bigl\{\underbrace{V(i-1,w)}_{\text{(a) leave item } i}, \underbrace{V(i-1,w-w_i)+v_i}_{\text{(b) take item } i}\bigr\}
\end{align*}
for $1 \leq i \leq n, 0 \leq w \leq W$. Essentially, the optimal value $V(i,w)$ is achieved by making a binary decision for every item $i \in [n]$ relying on pre-calculated optimal solutions to sub-problems with items from $\{1, \ldots, i-1\}$. The options are (a) either leaving item $i$ where the optimal solution is realised by the maximum profit achieved with items $\{1, \ldots, i-1\}$ and capacity $w$. Option (b) deals with putting item $i$ into the knapsack (only possible if $w_i < w$) gaining profit $v_i$ at the cost of additional $w_i$ units of weight. In consequence, the optimal profit $V(i,w)$ is the optimal profit with items from $\{1, \ldots, i-1\}$ and capacity $w-w_i$, i.e., $V(i,w)=V(i-1, w-w_i)+v_i$. Initialization follows
\begin{alignat}{3}
& V(0,w) = 0 && \quad\quad \forall \, &&0 \,\leq w \leq W \label{eq:dp_opt_empty_knapsack} \\
& V(i,w) = -\infty && \quad\quad \forall \, && w < 0 \label{eq:dp_opt_invalid_solution}
\end{alignat}
which covers the base cases of an empty knapsack (Eq.~\eqref{eq:dp_opt_empty_knapsack}) and a negative capacity, i.e., invalid solution (Eq.~\eqref{eq:dp_opt_invalid_solution}), respectively. Eventually, $V(n,W)$ holds the profit of an optimal solution.

\begin{algorithm}[t]
\caption{DP-algorithm for \SHARPOPTKNAPSACK}
\label{alg:dp}
\SetKwInOut{Input}{Input}
\Input{Number of items $n$, capacity $W$}
\For{$w \gets 0$ \KwTo $W$}{
    $V(0, w) \gets 0$\; $C(0,w) \gets 1$\;
} 
\For{$i \gets 0$ \KwTo $n$}{
    $V(i, 0) \gets 0$\; $C(i, 0) \gets 1$\;
}
\For{$i \gets 1$ \KwTo $n$}{
    \For{$w \gets 1$ \KwTo $W$}{
        \uIf{$w_i > w$}{
            $V(i,w) \gets V(i-1,w)$\;
            $C(i,w) \gets C(i-1,w)$\;
        }\Else{
            \uIf{$V(i-1,w) = V(i-1,w-w_i)+v_i$}{
                $V(i,w) \gets V(i-1,w)$\;
                $C(i,w) \gets C(i-1,w) + C(i-1, w-w_i)$\;
            }\uElseIf{$V(i-1,w) > V(i-1,w-w_i)+v_i$}{
                $V(i,w) \gets V(i-1,w)$\;
                $C(i,w) \gets C(i-1,w)$\;
            }\Else{
                $V(i,w) \gets V(i-1,w-w_i)+v_i$\;
                $C(i,w) \gets C(i-1,w-w_i)$\;
            }
        }
    }
}
\Return{$V$, $C$}
\end{algorithm}

\begin{table}[htb]
\caption{Exemplary knapsack instance (left) and the dynamic programming tables $V(i,w)$ (center) and $C(i,w)$ (right) respectively. Table cells highlighted in \colorbox{gray!20}{light-gray} indicate components where two options are possible: either packing item $i$ or not. \colorbox{green!20}{Light-green} cells indicate the total value of any optimal solution for $V(\cdot,\cdot)$ and $|\KPOPT|$ for $C(\cdot, \cdot)$.}
\label{tab:exact_counting_example}
\renewcommand{\tabcolsep}{3.3pt}
\renewcommand{\arraystretch}{1.3}
\centering
\begin{footnotesize}
\begin{tabular}{|r|r|r|}
\multicolumn{3}{l}{\vspace{0.45cm}} \\
\multicolumn{3}{c}{$W=8$} \\
\hline
$i$ & $w_i$ & $v_i$ \\
\hline
1 & 3 & 3 \\
\hline
2 & 8 & 10 \\
\hline
3 & 2 & 3 \\
\hline
4 & 2 & 4 \\
\hline
5 & 2 & 3 \\
\hline
\end{tabular}
\hskip5pt
\begin{tabular}{|c|c|c|c|c|c|c|c|c|c|}
\multicolumn{10}{c}{$V(\cdot, \cdot)$} \\
\hline
\diagbox{$i$}{$w$}
& 0 & 1 & 2 & 3 & 4 & 5 & 6 & 7 & 8 \\
\hline
0 & 0 & 0 & 0 & 0 & 0 & 0 & 0 & 0 & 0 \\
\hline
1 & 0 & 0 & 0 & 3 & 3 & 3 & 3 & 3 & 3 \\
\hline
2 & 0 & \tikzmark{v21}0 & 0 & \tikzmark{v23}3 & 3 & 3 & 3 & 3 & 10 \\
\hline
3 & 0 & 0 & 3 &  \tikzmark{v34}\cellcolor{gray!20}{3} &\cellcolor{gray!20}{3} & 6 & 6 & 6 & 10 \\
\hline
4 & 0 & 0 & 4 & 4 & 7 & 7 & 7 & 10 & \cellcolor{gray!20}{10} \\
\hline
5 & 0 & 0 & 4 & 4 & \cellcolor{gray!20}{7} & \cellcolor{gray!20}{7} & 10 & 10 & \cellcolor{green!20}{10} \\
\hline
\end{tabular}
\hskip5pt
\begin{tabular}{|c|c|c|c|c|c|c|c|c|c|}
\multicolumn{10}{c}{$C(\cdot,\cdot)$} \\
\hline
\diagbox{$i$}{$w$}
& 0 & 1 & 2 & 3 & 4 & 5 & 6 & 7 & 8 \\
\hline
0 & 1 & 1 & 1 & 1 & 1 & 1 & 1 & 1 & 1 \\
\hline
1 & 1 & 1 & 1 & 1 & 1 & 1 & 1 & 1 & 1 \\
\hline
2 & 1 & 1 & 1 & 1 & 1 & 1 & 1 & 1 & 1 \\
\hline
3 & 1 & 1 & 1 & \cellcolor{gray!20}{2} & \cellcolor{gray!20}{2} & 1 & 1 & 1 & 1 \\
\hline
4 & 1 & 1 & 1 & 1 & 1 & 2 & 2 & 1 & \cellcolor{gray!20}{2} \\
\hline
5 & 1 & 1 & 1 & 1 & \cellcolor{gray!20}{2} & \cellcolor{gray!20}{3} & 1 & 3 & \cellcolor{green!20}{4} \\
\hline
\end{tabular}
\end{footnotesize}
\end{table}

\subsection{Dynamic Programming for \SHARPOPTKNAPSACK}

We observe that if $V(i-1,w) = V(i-1,w-w_i)+v_i$ we can achieve \emph{the same maximum profit} $V(i,w)$ by both options (a) or (b). Analogously to $V(i,w)$ let $C(i,w)$ be the number of solutions with maximum profit given items from $\{1, \ldots, i\}$ and capacity $w$. Then there are three update options for $C(i,w)$ for $1 \leq i \leq n$ and $0 \leq w \leq W$ (we discuss the base case later): (a') Either, as stated above, we can obtain the same maximum profit $V(i,w)$ by packing or not packing item $i$. In this case, $C(i,w)=C(i-1,w)+C(i-1,w-w_i)$ since the item sets leading to $V(i-1,w)$ and $V(i-1,w-w_i)$ are necessarily disjoint by construction. Options (b') and (c') correspond to (a) and (b) leading to the recurrence

\begin{footnotesize}
\begin{align}
C(i,w) =
\begin{cases}
C(i-1,w)+C(i-1,w-w_i) & \text{ if } V(i-1,w)=V(i-1,w-w_i)+v_i\\
C(i-1,w) & \text{ if } V(i-1,w)>V(i-1,w-w_i)+v_i\\
C(i-1,w-w_i) & \text{ otherwise.}
\end{cases}
\label{eq:dp_count_recurrence}
\end{align}
\end{footnotesize}
Analogously to Equations~\eqref{eq:dp_opt_empty_knapsack} and \eqref{eq:dp_opt_invalid_solution} the bases cases
\begin{align}
C(0,w)=1 & \quad \forall \, 0\leq w \leq W \label{eq:dp_count_empty_knapsack} \\
C(i,w)=0 & \quad \forall \, w < 0 \label{eq:dp_count_invalid_solution}
\end{align}
handle the empty knapsack (the empty-set is a valid solution) and the case of illegal items (no valid solution(s) at all).
The base cases are trivially correct. The correctness of the recurrence follows inductively by the preceding argumentation. Hence, we wrap up the insights in the following lemma whose proof is embodied in the preceding paragraph.
\begin{lemma}
\label{lem:dp_count_algorithm_table_correctness}
Following the recurrence in Eq.~\eqref{eq:dp_count_recurrence} $C(i,w)$ stores the number of optimal solutions using items from $\{1, \ldots, i\}$ given the capacity $w$.
\end{lemma}

Algorithm~\ref{alg:dp} shows pseudo-code for our algorithm. The algorithm directly translates the discussed recurrences for tables $V$ and $C$. Note that the pseudo-code is not optimised for performance and elegance, but for readability.

\begin{theorem}
\label{thm:dp_count_algorithm}
Let $\KPOPT$ be the set of optimal solutions to a zero-one knapsack problem with $n$ items and capacity $W$. There is a deterministic algorithm that calculates $|\KPOPT|$ exactly with time- and space-complexity $O(n^2W)$.
\end{theorem}
\begin{proof}
The correctness follows from Lemma~\ref{lem:dp_count_algorithm_table_correctness}. For the space- and time-complexity note that two tables with dimensions $(n+1) \times (W+1)$ are filled. For Table $V$ the algorithm requires constant time per cell and hence time and space $O(nW)$. Table $C$ stores the number of solutions which can be exponential in the input size $n$ as we shall see later. Therefore, $C(i, j) \leq 2^n$ and $O(\log(2^n)) = O(n)$ bits are necessary to encode these numbers. Thus, the addition in line~15 in Algorithm~\ref{alg:dp} requires time $O(n)$ which results in space and time requirement of $O(n^2W)$.\hfill\qedhere{}
\end{proof}

Note that the complexity reduces to $O(nW)$ if $|S^{*}| = poly(n)$. Furthermore, the calculation of row $i$ only relies on values in row $i-1$. Hence, we can reduce the complexity by a factor of $n$ if only two rows of $V$ and $C$ are stored and the only value of interest is the number of optimal solutions.


For illustration we consider a simple knapsack instance with $n=5$ items fully described by the left-most table in Table~\ref{tab:exact_counting_example}.
Let $W=8$. In this setting there exist four optima $s_1 = \{2\}$, $s_2 = \{1,4,5\}$, $s_3=\{1,3,4\}$ and $s_4=\{3,4,5\}$ with profit~10 each and thus $|\KPOPT|=4$. The dynamic programming tables are shown in Table~\ref{tab:exact_counting_example} (center and right). For improved visual accessibility we highlight table cells where both options (packing item $i$ or not) are applicable and hence an addition of the number of combinations of sub-problems is performed in $C$ by the algorithm (cf. first case in the recurrence in Eq.~\ref{eq:dp_count_recurrence}).

\subsection{Uniform Sampling of Optimal Solutions}

\begin{algorithm}[ht]
\caption{Uniform Sampling of Optima}
\label{alg:dp_sampler}
\SetKwInOut{Input}{Input}
\Input{DP tables $V$ and $C$ (see Algorithm~\ref{alg:dp}), number of items $n$, capacity $W$, desired number of solutions $k$.
}
$S \gets \emptyset$\; 
\While{$k > 0$}{
    $L \gets \emptyset$\; 
    $i \gets n$\;
    $w \gets W$\;
    \While{$i > 0$ \textbf{and} $w > 0$}{
        \uIf{$w_i \leq w$ $\wedge$ $V(i,w) = V(i-1,w)$ $\wedge$ $V(i,w) = V(i-1,w-w_i)+v_i$}{ 
            $q \gets \frac{C(i-1,w-w_i)}{C(i,w)}$\;
            Let $r$ be a random number in $[0,1]$\;
            \If{$r < q$}{
                $L \gets L \cup \{i\}$\;
                $w \gets w - w_i$\;
            }
        }\uElseIf{$V(i,w) > V(i-1,w)$}{ 
            $L \gets L \cup \{i\}$\;
            $w \gets w - w_i$\;
        }
        $i \gets i - 1$\; 
    }
    $S \gets S \cup \{L\}$\;
    $k \gets k - 1$\;
}
\Return{S}
\end{algorithm}

The DP algorithm introduced before allows to count $|\KPOPT|$ and hence to solve \SHARPOPTKNAPSACK exactly. The next logical step is to think about a sampler, i.e., an algorithm that samples uniformly at random from $\KPOPT$ even if $\KPOPT$ is exponential in size. In fact, we can utilize the tables $C$ and $V$ for this purpose as they implicitly encode the information on all optima. A similar approach was used by Dyer~\cite{Dyer2003ApproximateCounting} to approximately sample from the (approximate) set of feasible solutions in his dynamic programming approach for \SHARPKNAPSACK. Our sampling algorithm however is slightly more involved due to a necessary case distinctions.

The algorithm starts at $V(n,W)$ respectively and reconstructs a solution $L \subset [n]$, initialized to $L=\emptyset$, bottom-up by making occasional random decisions. Assume the algorithm is at position $(i,w), 1 \leq i \leq n, 1 \leq w \leq W$. Recall (cf. Algorithm~\ref{alg:dp}) that if $V(i-1,w)<V(i-1,w-w_i)+v_i$ we have no choice and we need to put item $i$ into the knapsack. Likewise, if $V(i-1,w)>V(i-1,w-w_i)+v_i$, item $i$ is guaranteed not to be part of the solution under reconstruction. Thus, in both cases, the decision is deterministic. If $V(i-1,w)$ equals $V(i-1,w-w_i)+v_i$, there are two options how to proceed: in this case with probability $$\frac{C(i-1,w-w_i)}{C(i,w)}$$ item $i$ is added to $L$ and with the converse probability $$\frac{C(i-1,w)}{C(i,w)} = 1 -\frac{C(i-1,w-w_i)}{C(i,w)}$$ item $i$ is ignored. If $i$ was packed, the algorithm proceeds (recursively) from $V(i-1,w-w_i)$ and from $V(i-1,w)$ otherwise. This process is iterated while $i>0$ and $w>0$. To sample $k$ solutions we may repeat the procedure $k$ times which results in a runtime of $O(kn)$. This is polynomial as long as $k$ is polynomially bounded and of benefit for sampling from an exponential-sized set $\KPOPT$ if $W$ is low and hence Algorithm~\ref{alg:dp} runs in polynomial time. 
Detailed pseudo-code of the procedure is given in Algorithm~\ref{alg:dp_sampler}.

We now show that Algorithm~\ref{alg:dp_sampler} in fact samples each optimal solution uniformly at random from $\KPOPT$.

\begin{theorem}
Let $\KPOPT$ be the set of optimal solutions for the knapsack problem and $s \in \KPOPT$ be an arbitrary optimal solution. Then the probability of sampling $s$ using the sampling approach is $1/|\KPOPT|$.
\end{theorem}

\begin{proof}
Let $s \in \KPOPT$ be an optimal solution. Note that after running Algorithm~\ref{alg:dp} we have $C(n,W)=|\KPOPT|$. For convenience we assume that $C(i, w)=1$ for $w < 0$ to capture the case of invalid solutions. Consider the sequence of $1 \leq r \leq n$ decisions made while traversing back from position $(n,W)$ until a termination criterion is met (either $i \leq 0$ or $w \leq 0$) in Algorithm~\ref{alg:dp_sampler}. Let $q_i = \frac{a_i}{b_i}, i \in [r]$ be the decision probabilities in iterations $i \in [r]$. Here, $q_i$ corresponds to $q$ in line~8 of Algorithm~\ref{alg:dp_sampler} if there is a choice whether the corresponding item is taken or not. If there is no choice we can set $q_i = 1 = \frac{x}{x}$ with $x=C(i-1,w)$ if the item is not packed and $x=C(i-1,w-w_i)$ if the item is packed. A key observation is that (1) $b_1=C(n,W)$, (2) $b_{i}=a_{i-1}$ holds for $i = 2, \ldots, r$ by construction of $C$ and (3) $a_r=1$ since the termination condition applies after $r$ iterations (see base cases for $C$ in Eq.~\ref{eq:dp_count_empty_knapsack} and Eq.~\ref{eq:dp_count_invalid_solution}). Hence, the probability to obtain $s$ is
\begin{align*}
    \prod_{i=1}^{r} q_i
    & = \frac{a_1}{b_1} \cdot \frac{a_2}{b_2} \cdot \ldots \cdot \frac{a_{r-1}}{b_{r-1}} \cdot \frac{a_r}{b_r} \\
    & = \frac{a_1}{b_1} \cdot \frac{a_2}{a_1} \cdot \ldots \cdot \frac{a_{r-1}}{a_{r-2}} \cdot \frac{1}{a_{r-1}} \\
    & = \frac{1}{b_1} \\
    & = \frac{1}{C(n, W)} \\
    & = \frac{1}{|\KPOPT|}.
\end{align*}\hfill\qedhere{}
\end{proof}

\begin{theorem}
\label{thm:dp_sampler_samples_uniformly}
Let $\KPOPT$ be the set of optimal solutions for the knapsack problem with $n$ items. Algorithm~\ref{alg:dp_sampler} samples $k$ uniform samples of $\KPOPT$ in time $O(kn)$.
\end{theorem}

\begin{proof}
The probabilistic statement follows Theorem~\ref{thm:dp_sampler_samples_uniformly}. For the running time we note that each iteration takes at most $n$ iterations each with a constant number of operations being performed. This is repeated $k$ times which results in $O(kn)$ runtime which completes the proof.\hfill\qedhere{}
\end{proof}

\section{Experiments}
\label{sec:sec04}

In this section we complement the preceding sections with an experimental study and some derived theoretical insights. We first detail the experimental setup, continue with the analysis and close with some remarks.

\subsection{Experimental Setup}
The main research question is to get an impression and ideally to understand how the number of global optima develops for classical KP benchmark instances dependent on their generator-parameters. To this end we consider classical random KP benchmark generators as studied, e.g., by Pisinger in his seminal paper on hard knapsack instances~\cite{Pisinger2005HardKnapsackProblems}. 
\begin{figure*}[t]
    \centering
    \includegraphics[width=1.0\textwidth]{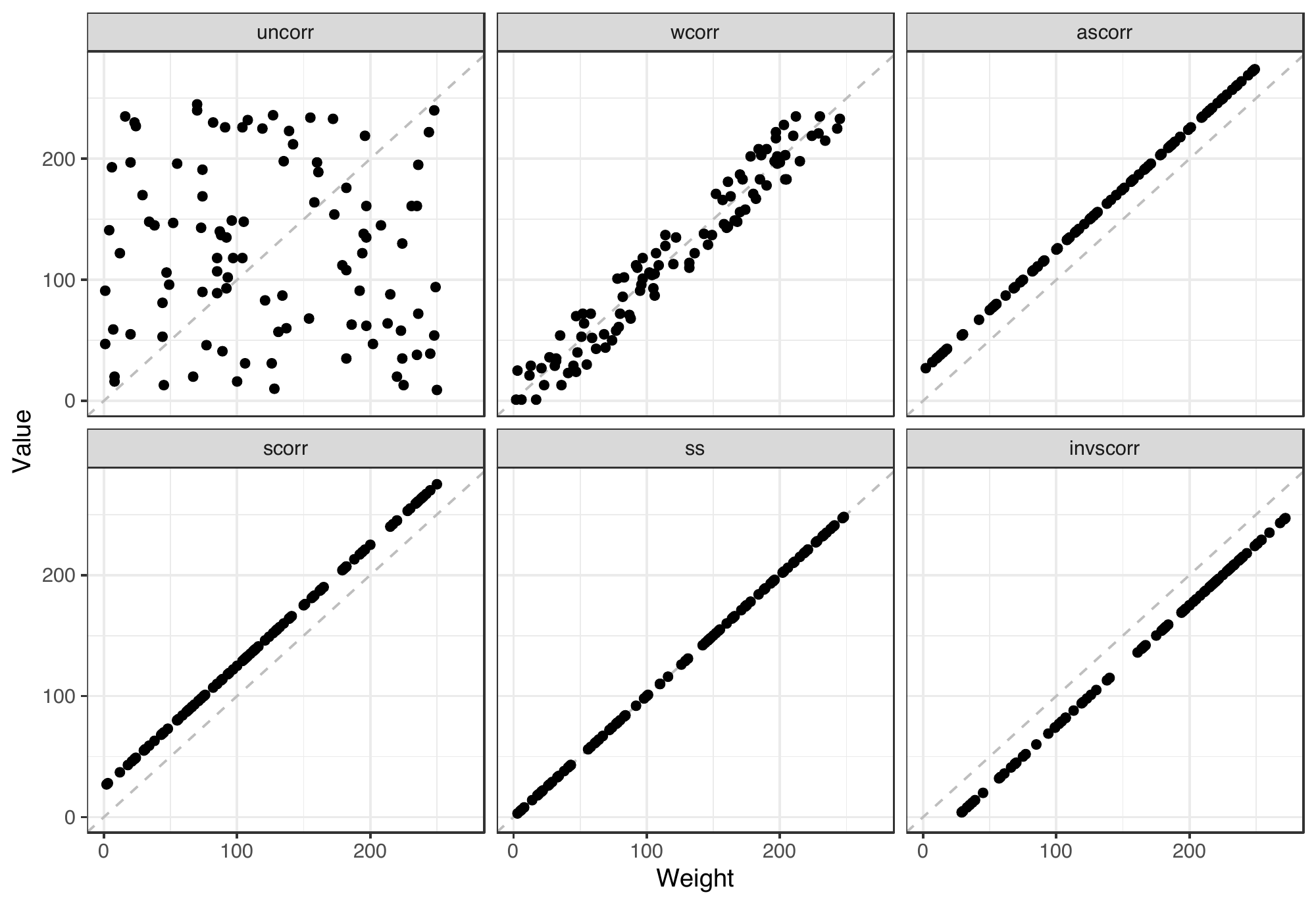}
    \caption{Showcase of considered instance groups. We show each one instance with $n=100$ items, $L=1$ and $R=250$. The gray dashed diagonal lines serves to aid recognising the differences between correlated instance groups.}
    \label{fig:instance_showcase}
\end{figure*}
All considered instance groups are generated randomly with item weights $w_i$ sampled uniformly at random within the data range $\{L, \ldots, R\}$ with $L=1$ and varying $R$ for $i \in [n]$; in any case $L < R$. Item profits $v_i$ are in most cases based on a mapping of item weights. The reader may want to take a look at  Figure~\ref{fig:instance_showcase} alongside the following description for visual aid.

\begin{description}
\item[Uncorrelated (uncorr)] Here, both weights and profits are sampled uniformly at random from $\{L, \ldots, R\}$.
\item[Weakly correlated (wcorr)] Weights $w_i$ are distributed in $\{L, \ldots, R\}$ and profits $v_i$ are sampled from $[w_i - R/10, w_i + R/10]$ ensuring $v_i \geq 1$. Here, the values are typically only a small percentage off the weights.
\item[Almost strongly correlated (ascorr)] 
Weights are distributed in $\{L, \ldots, R\}$ and $v_i$ are sampled from $[w_j + R/10 - R/500, w_j + R/10 + R/500]$.
\item[Strongly correlated (scorr)] Weights $w_i$ are uniformly distributed from the set $\{L, \ldots, R\}$ while profits are corresponding to $w_i + R/10$. Here, for all items the profit equals the positive constant plus some fixed additive constant.
\item[Subset sum (susu)] In this instance group we have $w_i = v_i \, \forall i \in [n]$, i.e., the profit equals the weight. This corresponds to strong correlation with additive constant of zero.
\item[Inversely strongly correlated (invscorr)] Here, first the profits $v_i$ are sampled from $\{L, \ldots, R\}$ and subsequently weights we set $w_i=v_i + R/10$. This is the counterpart of strongly correlated instances.
\end{description}
Correlated instances may seem highly artificial on first sight, but they are of high practical relevance. In economics this situation arises typically if the profit of an investment is directly proportional to the investment plus some fixed charge (strongly correlated) with some problem-dependent noise (weakly / almost strongly correlated). 

In our experiments we vary the instance group, the number of items $n \in \{50, 100, \ldots, 500\}$ and the upper bound $R \in \{25, 50, 100, 500\}$. In addition, we study different knapsack capacities by setting $D=11$ and
$$
W = \left\lfloor \frac{d}{D+1} \sum_{i=1}^{n} w_i\right\rfloor
$$
for $d = 1, \ldots, D$~\cite{Pisinger1999CoreProblemsKnapsack,Pisinger2005HardKnapsackProblems}. Intuitively -- for most considered instance groups -- the number of optima is expected to decrease on average for very low and very high capacities as the number of feasible/optimal combinations is likely to decrease.
For each combination of these parameters, we construct 25 random instances and run the DP algorithm to count the number of optima.

Python~3 implementations of the algorithms and generators and the code for running the experiments can be downloaded from a public GitHub repository.\footnote{\url{https://github.com/jakobbossek/LION2021-knapsack-exact-counting}} The experiments were conducted on a MacBook Pro 2018 with a 2,3 GHz Quad-Core Intel Core i5 processor and 16GB RAM. The operating system was macOS Catalina 10.15.6 and python v3.7.5 was used. Random numbers were generated with the built-in python module \texttt{random} while \texttt{joblist} v0.16.0 served as a multi-core parallelisation backend.

\subsection{Insights into the Number of Optima}

\begin{figure*}[htbp]
    \centering
    \includegraphics[width=\textwidth]{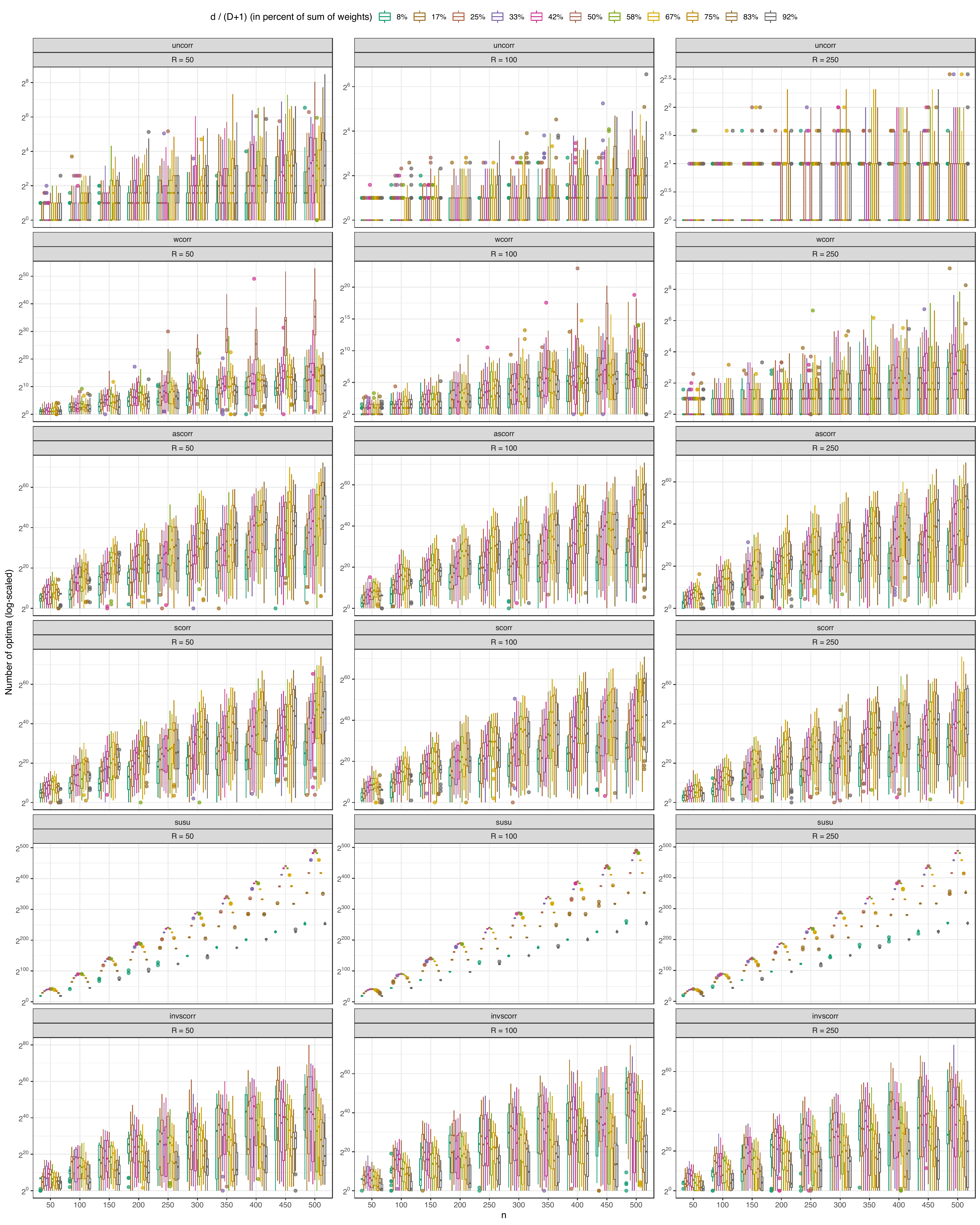}
    \caption{Boxplots of the number of optima as a function of the number of items $n$. The data is split by instance group (rows) and the upper bound $R$ (columns). Different colors indicate the knapsack capacity (shown as percentage of the sum of all items weights).}
    \label{fig:number_of_optima_distribution}
\end{figure*}

Figure~\ref{fig:number_of_optima_distribution} depicts the distribution of the number of optima for each combination considered in the experiments via boxplots. The data is split row-wise by instance group and col-wise by $R$ (recall that $L=1$ in any case). Different box-colors indicate the knapsack capacity which -- for sake of interpretability -- is given in percentage of the total weight, i.e., $\lfloor 100 \cdot (\frac{d}{D+1})\rfloor$.
We observe different patterns dependent on the instance group.
For uncorrelated instances (uncorr), we observe only few optima with single outliers reaching $2^8 = 256$ for $R=50$ and large $n$. Median values are consistently below $2^4=16$. In line with expectation the numbers are highest for relatively small $R$ and high $n$. In fact, the ratio $$H=\frac{n}{R}$$ is a good indicator. It is the expected number of elements with weight $w=1,\ldots,R$. In consequence, $H>1$ and especially $H\gg1$ indicates many elements with the same weight. In contrast, $H<1$ indicates that on average there will be at most one element of weight $w=1,\ldots,R$.
For all correlated instances, i.e., scorr, ascorr, wcorr, invscorr and susu, we observe a very different pattern. Here, the number of optima grows exponentially with growing $n$ given a fixed upper bound $R$. Even if $H$ is low there is huge number of optima. By far the highest count of optima can be observed for subset sum (susu) instances where even peaks with up to $\approx 3\%$ of all $2^n$ solutions are optimal. Here, the boxplots look degenerate, because the variance is very low. Recall that for this type of instance we have $w_i=v_i \, \forall i \in [n]$ and thus for each solution $s$ the equality $w(s)=v(s)$ holds. In consequence we aim to maximally exploit the knapsack capacity. 

To get a better understanding we consider a subset-sum type knapsack instance with $w_i \in \{1, \ldots, R\}, w_i=v_i, \forall i \in [n]$. Assume for ease of calculations that $n$ is a multiple of $R$ and there are exactly $(n/R)$ items of each weight $w \in \{1, \ldots, R\}$, i.e., $|\{i \in [n] \,|\, w_i = w\}| = n/R$. Note that this corresponds to the expected number of $w$-weights if $n$ such weights are sampled uniformly at random from $\{1, \ldots, R\}$. Consider $W = \frac{1}{2} \sum_{i=1}^{n} w_i$. Recall that given this instance, one way we can build an optimum $s \subset [n]$ with $w(s)=v(s)=W$ is by choosing each $\left(\frac{1}{2}\right) \cdot \left(\frac{n}{R}\right)$ items from each weight class, i.e., half of these items (note that there are many more combinations leading to profit $W$). With this we get
\begin{align*}
|\KPOPT| 
\geq \binom{\frac{n}{R}}{\frac{n}{2R}}^R 
\geq \left(\left(\frac{n}{R} \cdot \frac{2R}{n}\right)^{\frac{n}{2R}}\right)^R
= 2^{\frac{n}{2}}.
\end{align*}
Here we used to the well-known lower bound $\binom{n}{k} \geq \left(\frac{n}{k}\right)^k$ for the binomial coefficient. This simple bound establishes that we can expect at least $2^{n/2}$ optima for subset-sum instances if the capacity is set accordingly.

With respect to the knapsack capacity Figure~\ref{fig:number_of_optima_distribution} also reveals different patterns. For inverse strongly correlated instances we observe a decreasing trend with increasing capacity. The vice versa holds for weakly, almost strongly and strongly correlated instances. This is in line with intuition as the size of the feasible search space also grows significantly.

However, note that in general the knapsack capacity can have a massive effect on the number of optima. 

\begin{theorem}
\label{thm:existence_instance_sensitive_capacity}
For every even $n$ there exist a KP instance and a weight capacity $W$ such that $|\KPOPT|$ is exponential, but $|\KPOPT|=1$ for $W'=W+1$.
\end{theorem}

\begin{proof}
Consider an instance with $n$ items ($n$ even), where $w_i = v_i = 1$ for $i \in [n-1]$ and $w_n = \frac{n}{2}+1$. Let $v_n > \frac{n}{2}+1$. Now consider the knapsack capacity $W = \frac{n}{2}$. Then every subset of $\frac{n}{2}$ items from the first $n-1$ items is  optimal with total weight $W$ and total value $W$ while the $n$-th item does not fit into the knapsack. There are at least
\begin{align*}
    \binom{n-1}{\frac{n}{2}} 
    & = \frac{(n-1)!}{(n/2)!(n-1-n/2)!} \\
    & \geq \frac{(n-1)!}{\left(\left(\frac{n}{2}\right)!\right)^2} \\
    & \geq \frac{\sqrt{2\pi} \sqrt{n-1} (n-1)^{n-1} 2^n e^n}{2 e^{n-1} 4 (2\pi) n^{n+1}} \\
    & = \frac{2^{n-3}e}{\sqrt{2\pi}} \cdot \underbrace{\left(\frac{n-1}{n}\right)^{n-1}}_{\geq e^{-1}} \cdot \frac{\sqrt{n}}{n^2} \cdot \underbrace{\sqrt{1-\frac{1}{n}}}_{\geq 1/\sqrt{2}} \\
    & \geq \frac{1}{\sqrt{\pi}} \cdot \frac{2^{n-4}}{n^{3/2}} \\ 
    & = \Omega(2^{n-4}/n^{3/2})
\end{align*}
optima in this case. Here we basically used Stirling's formula to lower/upper bound the factorial expressions to obtain an exponential lower bound.
Now instead consider the capacity $W'=W+1$. The $n$-th item now  fits into the knapsack which results in a unique optimum with weight $W'=\frac{n}{2}+1$ and value $v_n > \frac{n}{2}+1$ which cannot be achieved by any subset of light-weight items.\hfill\qedhere{}
\end{proof}

\subsection{Closing Remarks}

Knapsack instances with correlations between weights and profits are of high practical interest as they arise in many fixed charge problems, e.g., investment planning. In this type of instances item profits correspond to their weight plus/minus a fixed or random constant. Our experimental study suggests an exponential increase in the number of global optima for such instances which justifies the study and relevance of the considered counting and sampling problem.

\section{Conclusion}
\label{sec:sec05}

We considered the problem of counting exactly the number of optimal solutions of the zero-one knapsack problem. We build upon the classic dynamic programming algorithm for the optimisation version. Our modifications allow to solve the counting problem in pseudo-polynomial runtime complexity. Furthermore, we show how to sample uniformly at random from the set of optimal solutions without explicit construction of the whole set. Computational experiments and derived theoretical insights reveal that for variants of problem instances with correlated weights and profits (a group which is highly relevant in real-world scenarios) and for a wide range of problem generator parameters, the number of optimal solutions can grow exponentially with the number of items. These observations support the relevance of the considered counting and sampling problems.

Future work will focus on (approximate and exact) counting/sampling of high-quality knapsack solutions which all fulfill a given non-optimal quality threshold. In addition, in particular if the set of optima has exponential size, it is desirable to provide the decision maker with a diverse set of high-quality solutions. Even though the introduced sampling is likely to produce duplicate-free samples if the number of solutions is exponential, it seems more promising to bias the sampling process towards a diverse subset of optima, e.g., with respect to item-overlap or entropy. This opens a whole new avenue for upcoming investigations.

\bibliographystyle{splncs04}
\bibliography{lion.bib}

\end{document}